\documentclass[a4paper,12pt]{article}
\usepackage[utf8x]{inputenc}
\usepackage{amssymb,amsmath,amsthm,amstext,graphics,amsfonts,hyperref}
\DeclareMathOperator{\cwe}{cwe}
\DeclareMathOperator{\swe}{swe}
\DeclareMathOperator{\swc}{swc}
\DeclareMathOperator{\wtg}{wt_L}
\DeclareMathOperator{\wtgr}{Wt_L}
\DeclareMathOperator{\ord}{Ord}

\def\tilam{\phi_{S_1,S_2}}

\newtheorem{theorem}{Theorem}[]
\newtheorem{lemma}[theorem]{Lemma}

\newtheorem{corollary}[theorem]{Corollary}
\newtheorem{proposition}[theorem]{Proposition}
\theoremstyle{definition}
\newtheorem{example}[theorem]{Example}
\newtheorem{definition}[theorem]{Definition}

\newtheorem{alg}[theorem]{Algorithm}
\title{Codes over An Algebra over Ring}
\author{Irwansyah\\
{\it \small Department of Mathematics,}\\
{\it \small Faculty of Mathematics and Natural Sciences,}\\
{\it \small Universitas Mataram, Jl. Majapahit 62, Mataram, 83125}\\
{\it \small INDONESIA}\\
\vspace{0.1cm}\\
Djoko Suprijanto\\
{\it \small Combinatorial Mathematics Research Group,}\\
{\it \small Faculty of Mathematics and Natural Sciences,}\\
{\it \small Institut Teknologi Bandung,
Jl. Ganesha 10, Bandung, 40132,}\\
{\it \small INDONESIA}}
\date{}

\begin{document}
\maketitle
\begin{abstract}
In this paper, we consider some structures of linear codes over the ring $\mathcal{R}_k=R[v_1,\dots,v_k],$ where $v_i^2=v_i$ forall $i=1,\dots,k),$ and $R$ is a
finite commutative Frobenius ring.
\\[0.25cm]
\\
%\textbf{11T71}
\end{abstract}

\section{Introduction}

Some special cases of codes over the ring of the form $\mathcal{R}_k=R[v_1,\dots,v_k],$ where $v_i^2=v_i$ for all $i=1,\dots,k,$ and $R$ is a finite commutative
Frobenius ring
attract the attention of some researchers in coding theory. This is because codes over such kind of rings have a lot of nice structures. For example,
in \cite{abualrub,gao,irw}, they consider skew-cyclic codes over the ring $\mathbb{F}_2+v\mathbb{F}_2,\mathbb{F}_p+v\mathbb{F}_p$ and
$\mathbb{F}_{p^r}[v_1,\dots,v_k],$
respectively. Moreover, in \cite{dougherty-ceng,gao2,gao3}, they studied the structures of codes over $\mathbb{F}_{2}[v_1,\dots,v_k],\mathbb{Z}_4+v\mathbb{Z}_4,$ and
$\mathbb{Z}_9+v\mathbb{Z}_9,$ respectively, such as MacWilliams identity, self-dual codes, cyclic codes, constacyclic codes, {\it etc}. Also, we can find
a constructon of good and new $\mathbb{Z}_4$-linear codes in \cite{gao2}.

In this paper, we try to give general recipes for the structures of codes over such class of rings, including MacWilliams identities,
self-dual codes, cyclic codes, quasi-cyclic
codes, skew-cyclic codes, and quasi-skew-cyclic codes.

\section{Automorphisms and Gray Map}\label{grayandauto}

Let $R$ be a finite Frobenius ring and $\mathcal{R}_k=R[v_1,v_2,\dots,v_k],$ for some $k\in\mathbb{N},$ where $v_i^2=v_i,$ for all $i=1,2,\dots,k.$
The ring $\mathcal{R}_k$
can be viewed
as a free module over $R$ with dimension $2^k.$ Let $w_i=\{1-v_i,v_i\}$ and $w_S=\prod_{i\in S}w_i.$ Then, we have the following immediate properties.

\begin{lemma}
The ring $\mathcal{R}_k$ has the cardinality $|R|^{2^k}$ and characteristic equals to $char(R).$
\end{lemma}

\begin{proof}
As we can see, every element $\alpha\in\mathcal{R}_k$ can be written as
\[\alpha = \sum_{i=1}^{2^k}\alpha_{S_i}v_{S_i},\]
for some $\alpha_{S_i}\in R,$ for all $1\leq i\leq 2^k.$ Therefore we have that $|\mathcal{R}_k|=|R|^{2^k}.$
\end{proof}

Let $\Theta_i$ be a map on $\mathcal{R}_k$ such that
\[\Theta_i(\alpha)=\left\{\begin{array}{ll}
                    1-v_i, & \text{if}\;\alpha=v_i\\
                    \alpha, & \text{otherwise}.
                   \end{array}\right.\]
Then define
\[\Theta_S=\prod_{i\in S}\Theta_i=\Theta_{i_1}\circ \Theta_{i_2}\circ\cdots\circ\Theta_{i_{|S|}},\]
where $S\subseteq\{1,2,\dots,k\}.$ \\[0.25cm]
Also, let $S_1,S_2\subseteq\{1,2,\dots,k\},$ where $|S_1|=|S_2|,$ and $\phi_{S_1,S_2} : \{1,2,\dots,k\}\rightarrow \{1,2,\dots,k\}$ be a map such that
it is a bijection from $S_1$ to $S_2$ and $\phi_{S_1,S_2}(j)=j,$ for all $j\not\in S_1.$ Define a map $\Phi_{S_1,S_2},$
where
\[\Phi_{S_1,S_2}(\alpha v_j)= \vartheta(\alpha)v_{\phi_{S_1,S_2}(j)},\]
for some  automorphism $\vartheta$ of $R.$

We have to note that the maps $\Theta_S$ and $\Phi_{S_1,S_2}$ are automorphisms on the ring $\mathcal{R}_k,$ so does their compositions. In
this paper we consider automorphism $\theta$ as a composition of $\Theta_S$ or $\Phi_{S_1,S_2}$ or both.

Now, we will define two Gray maps from the ring $\mathcal{R}_k.$ {\it First}, For any $j\geq 1,$ any element $\alpha$ in $\mathcal{R}_j$ can be written as $\alpha=\alpha_1+\alpha_2 v_j,$ for some
$\alpha_1,\alpha_2\in \mathcal{R}_{j-1}.$ For some $l_j\geq 2$ in $\mathbb{N},$ define a map

\[\begin{array}{llll}
\varphi_{j} : & \mathcal{R}_j & \longrightarrow & \mathcal{R}_{j-1}^{l_j} \\
& \alpha_1+\alpha_2 v_j & \longmapsto & \left(\alpha_1,\beta_1\alpha_1+\beta_1'\alpha_2,\beta_2\alpha_1+\beta_2'\alpha_2,\dots,\beta_{l_j-1}\alpha_1+\beta_{l_j-1}'\alpha_2\right).
\end{array}
\]
where $\beta_i,\beta_i'$ are some elements in $\mathcal{R}_{j-1},$ for all $1\leq i\leq l_j,$ with $\beta_{l_{j-1}}'$ is a unit in $\mathcal{R}_{j-1}.$ 
The following lemma shows that $\varphi_j$ is an injective map and also a module homomorphism.

\begin{lemma}\label{lemmagray}
	The map $\varphi_j$ is an injective and also a $\mathcal{R}_{j-1}$-module homomorphism from $\mathcal{R}_j$ to $\mathcal{R}_{j-1}^{l_j},$ for all $1\leq j\leq k.$
\end{lemma}

\begin{proof}
	For {\it injectivity}, take any $\alpha$ and $\alpha'$ in $\mathcal{R}_j,$ where $\varphi_j(\alpha)=\varphi_j(\alpha').$ Now, let $\alpha=\alpha_1+\alpha_2v_j$ and $\alpha'=\alpha_1'+\alpha_2'v_j,$ for some $\alpha_1,\alpha_2,\alpha_1',$ and $\alpha_2'$ in $\mathcal{R}_{j-1}.$ Since $\varphi_j(\alpha)=\varphi_j(\alpha'),$ we have $\alpha_1=\alpha_1'.$ Using the previous fact and by considering the last coordinate of the images under $\varphi_j,$ we have $\beta_{l_j}'\alpha_2=\beta_{l_j}'\alpha_2'.$ Since $\beta_{l_j}'$ is a unit in $\mathcal{R}_{j-1},$ we also have $\alpha_2=\alpha_2'$ as we hope.\\
	
	Now, take any $\alpha$ and $\alpha'$ in $\mathcal{R}_j$ and any $\lambda$ in $\mathcal{R}_{j-1}.$ Let $\alpha=\alpha_1+\alpha_2 v_j$ and $\alpha'=\alpha_1'+\alpha_2' v_j,$ for some $\alpha_1,\alpha_2,\alpha_1'$ and $\alpha_2'$ in $\mathcal{R}_{j-1}.$ Consider 
	\[\begin{array}{lll}
		\varphi_j(\alpha+\alpha') & = & \left(\alpha_1+\alpha_1',\beta_1(\alpha_1+\alpha_1')+\beta_1'(\alpha_2+\alpha_2'),\beta_2(\alpha_1+\alpha_1')+\beta_2'(\alpha_2+\alpha_2'),\dots\right. \\
		&  & \left.\dots,\beta_{l_j-1}(\alpha_1+\alpha_1')+\beta_{l_j-1}'(\alpha_2+\alpha_2')\right) \\
		& = & \varphi_j(\alpha)+ \varphi_j(\alpha'),	
	\end{array}\]  
	and
	\[\begin{array}{lll}
		\varphi_j(\lambda\alpha) & = & \left(\lambda\alpha_1,\beta_1\lambda\alpha_1+\beta_1'\lambda\alpha_2,\beta_2\lambda\alpha_1+\beta_2'\lambda\alpha_2,\dots,\beta_{l_j-1}\lambda\alpha_1+\beta_{l_j-1}'\lambda\alpha_2\right)\\
		& = & \lambda\varphi_j(\alpha).
	\end{array}\]
	Therefore, the map $\varphi_j$ is a $\mathcal{R}_{j-1}$-module homomorphism for all $1\leq j\leq k.$
\end{proof}

 Note that, we can combine the maps $\varphi_j$ and $\varphi_{j-1}$ to get a map from $\mathcal{R}_j$ to $\mathcal{R}_{j-2}^{l_j\times l_{j-1}}$ as follows.
 
 \[\begin{array}{llll}
 \varphi_{j-1}\circ\varphi_j : & \mathcal{R}_j & \longrightarrow & \mathcal{R}_{j-2}^{l_j\times l_{j-1}}\\
  & \alpha_1+\alpha_2v_j & \longmapsto & \left(\varphi_{j-1}\left(\alpha_1\right),\varphi_{j-1}\left(\beta_1\alpha_1+\beta_1'\alpha_2\right),\varphi_{j-1}\left(\beta_2\alpha_1+\beta_2'\alpha_2\right),\dots\right.\\
  &  &  &\left.\dots,\varphi_{j-1}(\beta_{l_j-1}\alpha_1+\beta_{l_j-1}'\alpha_2)\right)
 \end{array}\]
 By doing this inductively, we will have a map $\varphi_1\circ \varphi_2\circ\cdots\circ\varphi_k$ from $\mathcal{R}_k$ to $R^{l_k\times l_{k-1}\times\cdots\times l_1}.$  
 
 We can extend the map $\varphi_j$ to get a map from $\mathcal{R}_j^n$ to $\mathcal{R}_{j-1}^{nl_j}$ by the following way,
 
 \[\begin{array}{llll}
 	\overline{\varphi}_j : & \mathcal{R}_j^n & \longrightarrow & \mathcal{R}_{j-1}^{nl_j}\\
 	 & (\alpha_{1,1}+\alpha_{1,2}v_j,\dots,\alpha_{n,1}+\alpha_{n,2}v_j) & \longmapsto & \left(\alpha_{1,1},\dots,\alpha_{n,1},\beta_1\alpha_{1,1}+\beta_1'\alpha_{1,2},\right.\\
 	 &&&\dots,\beta_1\alpha_{n,1}+\beta_1'\alpha_{n,2},\dots \\
 	 & & & \dots,\beta_{l_j-1}\alpha_{1,1}+\beta_{l_j-1}'\alpha_{1,2},\dots\\
 	 &&&\left.\dots,\beta_{l_j-1}\alpha_{n,1}+\beta_{l_j-1}'\alpha_{n,2}\right) 
 \end{array}\]
 
 We can combine $\overline{\varphi}_j$  and $\overline{\varphi}_{j-1}$  to get a map from $\mathcal{R}_j^n$ to $\mathcal{R}_{j-2}^{nl_jl_{j-1}},$ and inductively, to get a map from $\mathcal{R}_k^n$ to $R^{nl_k\cdots l_1}.$ The map $\varphi_j$ and its extensions are a generalization of Gray maps in \cite{dougherty-ceng,irw}. 
 
 For the {\it second} Gray map, any $\alpha$ in $\mathcal{R}_k$ can be written as $\alpha = \sum_{i=1}^{2^k}\alpha_{S_i}v_{S_i},$ for some $\alpha_{S_i}$ in $R,$ where $S_i\subseteq \{1,2,\dots,k\}$ and $v_{S_i}=\prod_{t\in S_i}v_t,$ for all $1\leq i\leq 2^k.$ Define a map $\Psi$ as follows.
 \[\begin{array}{llll}
 \Psi : & \mathcal{R}_k & \longrightarrow & R^{2^k}\\
 & \sum_{i=1}^{2^k}\alpha_{S_i}v_{S_i} & \longmapsto & \left(\sum_{S\subseteq S_1}\alpha_S,\dots,\sum_{S\subseteq S_{2^k}}\alpha_S\right)
 \end{array}
 \]
 We can check that the map $\Psi$ is a bijection map. Moreover, we can also check that the map $\Psi$ is an isomorphism, which
 implies
 \[\mathcal{R}_k\cong \underbrace{R\times R\times\cdots\times R}_{2^k}.\]
 This means $\mathcal{R}_k$ is also a Frobenius ring.
 
 Let $\overline{\Psi} : \mathcal{R}_k^n\rightarrow R^{2^k\times n}$ be a map such that
 \[\overline{\Psi}(a_1,\dots,a_n)=\left(\Psi(a_1),\dots,\Psi(a_n)\right).\]
 Then, we can see that $\overline{\Psi}$ is also a bijective map because $\Psi$ is bijective. Let $\Sigma_S$ and $\Gamma_{S_1,S_2}$ be two maps such that $\overline{\Psi}\circ \Theta_S=\Sigma_S\circ\overline{\Psi}$ and
 $\overline{\Psi}\circ\Phi_{S_1,S_2}=\Gamma_{S_1,S_2}\circ \overline{\Psi}.$ As we can see, the maps $\Sigma_S$ and $\Gamma_{S_1,S_2}$ are bijective maps induced
 by $\Theta_S$ and $\Phi_{S_1,S_2},$ respectively.

\section{Linear and Self-Dual Codes}

In this part, we will describe linear codes over $\mathcal{R}_k$ using the gray map defined in Section~\ref{grayandauto}. The following theorems describe the image of a linear code under the gray maps $\overline{\varphi}_j$ and $\overline{\Psi}.$ The following theorem describe the image of a linear code under the map $\overline{\varphi}_j.$

\begin{theorem}
	A code $C$ is a linear code of length $n$ over $\mathcal{R}_j$ if and only if the image $\overline{\varphi}_j(C)$ is a linear code of length $nl_j$ over $\mathcal{R}_{j-1}.$
\end{theorem}
We have the following consequence.

\begin{corollary}
A code $C$ is a linear code of length $n$ over $\mathcal{R}_k$ if and only if the code
\[\overline{\varphi}_1\circ\overline{\varphi}_2\circ\cdots\circ\overline{\varphi}_k(C)\]
is a linear code of length $nl_1\cdots l_k$ over $R.$
\end{corollary}

The following theorem describe the image of a linear code under the map $\overline{\Psi}.$

\begin{theorem}\label{linearpsi}
		A code $C$ is a linear code of length $n$ over $\mathcal{R}_k$ if and only if there exist linear codes, $C_1,C_2,\dots,C_{2^k},$ of length $n$ over $R$ such that $C=\overline{\Psi}^{-1}(C_1,C_2,\dots,C_{2^k}).$
\end{theorem}

\begin{proof}
	Similar to the proof of \cite[Lemma 16]{irw}.
\end{proof}

Now, we will describe Euclidean and Hermitian self-dual codes. Let $\Theta_S$ be an automorphism in the ring $\mathcal{R}_k$ as in Section~\ref{grayandauto}, where $S=\{1,2,\dots,k\}.$
For any $\mathbf{c}=(c_1,\dots,c_n)$ and $\mathbf{c}'=(c_1',\dots,c_n')$ in $\mathcal{R}_k^n,$
define the Hermitian product as follows,
\[[\mathbf{c},\mathbf{c}']=\sum_{i=1}^n c_i\overline{c_i'}=\sum_{i=1}^nc_i\Theta_S(c_i').\]
Let
$C^H=\{\mathbf{c}'|[\mathbf{c},\mathbf{c}']=0\;\forall \mathbf{c}\in C\},$ then a code $C$ is called {\it Hermitian self-orthogonal} if
$C\subseteq C^H,$ and $C$ is called {\it Hermitian self-dual} if $C=C^H.$ Also, for any $\mathbf{c}=(c_1,\dots,c_n)$ and $\mathbf{c}'=(c_1',\dots,c_n'),$ define the Euclidean product as the following rational sum,
\[\mathbf{c}\cdot \mathbf{c}'=\sum_{i=1}^nc_ic_i'.\]
Let $C^\bot=\{\mathbf{c}'|\mathbf{c}\cdot \mathbf{c}'=0\;\forall \mathbf{c}\in C\},$ then a code $C$ is called {\it self-orthogonal} if
$C\subseteq C^\bot,$ and $C$ is called {\it Euclidean self-dual} if $C=C^\bot.$ 
The following theorem describe the existence of Hermitian self-dual codes over
$\mathcal{R}_k.$

\begin{theorem}
If $S\not=\emptyset,$ then there exist Hermitian self-dual codes over $\mathcal{R}_k$ for all length.
\label{hermitianconstruct}
\end{theorem}

\begin{proof}
 Take $i$ in $S.$ Let $C_1=\langle v_i\rangle,$ then we have $C_1^H=\langle v_i\rangle=C_1,$ because $v_i(1-v_i)=0.$ So, Hermitian self-dual code of length $1$ over
 $\mathcal{R}_k$ exist. Now, for any length $n,$ define
 \[C=\underbrace{C_1\times C_1\times\cdots\times C_1}_{n}.\]
 As we can see, $C^H=C,$ which means $C$ is an Hermitian self-dual code of length $n.$
\end{proof}

 Note that, the ring $\mathcal{R}_k$ can be written as $\mathcal{R}_k=v_k\mathcal{R}_{k-1}+(1-v_k)\mathcal{R}_{k-1}.$ Consequently, any code $C$ of length $n$ over $\mathcal{R}_{k}$ can be written as $C=v_k C_1+(1-v_k)C_2,$ where $C_1$ and $C_2$ are codes of length $n$ over $\mathcal{R}_{k-1}.$ 
 
\begin{proposition}
If $C$ is a Hermitian self-dual code of length $n$ over $\mathcal{R}_1,$ then $C$ is isomorphic to $C_1\times C_1^{\bot},$ where $C_1$ is a code
of length $n$ over $R.$
\label{hermitian}
\end{proposition}

\begin{proof}
Remember that $C$ can be written as $C=vC_1+(1-v)C_2,$ where $C_1$ and $C_2$ are codes of length $n$ over $R.$ Consider
\begin{equation}\label{hermiteq}
\begin{array}{lll}
[\mathbf{c},\mathbf{c}'] & = & \sum_i c_i\overline{c_i'}\\
 & = & \sum_i \left(vc_{1i} +(1-v) c_{2i}\right)\overline{\left(vc_{1i}' +(1-v) c_{2i}'\right)}\\
 & = & \sum_i \left(vc_{1i} +(1-v) c_{2i}\right)\left((1-v)c_{1i}' + vc_{2i}'\right)\\
 & = & v\sum_i c_{1i}c_{2i}' + (1-v)\sum_i c_{2i}c_{1i}',
\end{array}
\end{equation}
where $(c_{j1},c_{j2},\dots,c_{jn})$ is in $C_j,$ for $j=1,2.$ If the equation~\ref{hermiteq} is equal to $0,$ then it requires $\sum_i c_{1i}c_{2i}'=0$ and $\sum_i c_{2i}c_{1i}'=0.$ Since $C$ is self dual, we have $C_1=C_2^\bot$ and $C_2=C_1^\bot.$ Therefore, $C$ is isomorphic to $C_1\times C_1^\bot.$
\end{proof}

Using the above property, we have the following theorem.

\begin{theorem}
If $C$ is a Hermitian self-dual code of length $n$ over $\mathcal{R}_k,$ then, with proper arrangement of indices, $C$ is isomorphic to
\[C_1\times C_1^\bot\times\cdots\times C_{2^{k-1}}\times C_{2^{k-1}}^\bot,\]
where $C_1,\dots,C_{2^{k-1}}$ are codes of length $n$ over $R.$
\label{hermitian2}
\end{theorem}

\begin{proof}
We can write $C=v_kC'+(1-v_k)C'',$ where $C'$ and $C''$ are codes of length $n$ over $R_{k-1}.$ Consider 
\begin{equation}\label{hemiteq2}
\begin{array}{lll}
[\mathbf{c}_1,\mathbf{c}_2] & = & \sum_i c_{1i}\overline{c_{2i}}\\
 & = & \sum_i \left(v_kc_{1i}'+(1-v_k)c_{1i}''\right)\overline{\left(v_kc_{2i}'+(1-v_k)c_{2i}''\right)}\\
 & = & \sum_i \left(v_kc_{1i}'+(1-v_k)c_{1i}''\right)\left((1-v_k)\overline{c_{2i}'}+v_k\overline{c_{2i}''}\right)\\
 & = & v_k\sum_i c_{1i}'\overline{c_{2i}''} + (1-v_k)\sum_i c_{2i}'\overline{c_{1i}''},
\end{array}
\end{equation} 
where $(c_{j1}',c_{j2}',\dots,c_{jn}')$ is in $C'$ and $(c_{j1}'',c_{j2}'',\dots,c_{jn}'')$ is in $C'',$ for $j=1,2.$ If equation~\ref{hemiteq2} is $0,$ then it requires 
\begin{equation}\label{hermiteq3}
\sum_i c_{1i}'\overline{c_{2i}''}=0
\end{equation}
 and
\begin{equation}\label{hermiteq4}
\sum_i c_{2i}'\overline{c_{1i}''}=0.
\end{equation} 
If we continue similar process on equation~\ref{hermiteq3} and \ref{hermiteq4}, we will have $2^k$ equations similar to equation~\ref{hermiteq} over $R.$ By Proposition~\ref{hermitian}, $2^k$ equations give $2^{k-1}$ pairs of Euclidean dual over $R.$ Therefore, we have $C$ is isomorphic to 
\[C_1\times C_1^\bot\times\cdots\times C_{2^{k-1}}\times C_{2^{k-1}}^\bot,\]
where $C_1,C_2,\dots,C_{2^{k-1}}$ are codes of length $n$ over $R.$
\end{proof}

We have the following result.

\begin{theorem}
A code $C$ is an Euclidean self-dual code of length $n$ over $\mathcal{R}_k$ if and only if $C=\overline{\Psi}^{-1}(C_1,C_2,\dots,C_{2^k}),$ where $C_1,\dots,C_{2^k}$
are also Euclidean self-dual codes over $R.$
\label{euclid}
\end{theorem}

\begin{proof}
Similar to the proof of \cite[Proposition 4.1]{irw2}.
\end{proof}

We have the following immediate consequence.

\begin{corollary}
Eucidean self-dual codes of length $n$ over $\mathcal{R}_k$ exist if and only if Euclidean self-dual codes of length $n$ over $R$ exist.
\end{corollary}

\section{Weights and MacWilliams Identities}

Let $d_H(C)$ be the Hamming distance of a code $C.$ The following proposition gives the Hamming distance for codes over the ring $\mathcal{R}_k.$

\begin{proposition}
If $C=\overline{\Psi}^{-1}(C_1,\dots,C_{2^k}),$ is a code of length $n$ over $\mathcal{R}_k,$ then $d_H(C)=\min_{1\leq i\leq 2^k}d_H(C_i).$
\end{proposition}

\begin{proof}
Let $d_H(C_j)=\min_{1\leq i\leq 2^k}d_H(C_i),$ for some $j.$ Also, let $\mathbf{c}_j$ be a codeword in $C_j$ such that $wt(\mathbf{c}_j)=d_H(C_j).$ Then we have that
\[d_H(C)=wt\left(\overline{\Psi}^{-1}(\mathbf{0},\dots\mathbf{0},\mathbf{c}_j,\mathbf{0},\dots,\mathbf{0})\right)=d_H(C_j).\]
\end{proof}

Let $wt_H(\mathbf{c})$ be a {\it Hamming weight} of codeword $\mathbf{c}.$ Also, let
\[W_C(X,Y)=\sum_{\mathbf{c}\in C}X^{n-wt_H(\mathbf{c})}Y^{wt_H(\mathbf{c})},\]
be {\it the Hamming weight enumerator} of code $C.$ We have the following relation between Hamming weight enumerator of a code $C$ and its dual.

\begin{proposition}
If $C$ is a code of length $n$ over $\mathcal{R}_k,$ then
\[W_{C^\bot}(X,Y)=\frac{1}{|C|}W_C\left(X+(|R|^{2^k}-1)Y,X-Y\right).\]
\end{proposition}

\begin{proof}
Use the fact that $|\mathcal{R}_k|=|R|^{2^k}.$
\end{proof}

Now, let $\wtg(\alpha)$ be the Lee weight of any element $\alpha$ in $R.$ Let $a=\sum_{S\subseteq \{1,2,\dots,k\}}\alpha_Sv_S$ be any element in $\mathcal{R}_k.$
Define
\[\wtgr(a)=\sum_{i=1}^{2^k}\wtg\left(\sum_{S\subseteq S_i}\alpha_S\right)\]
be the Lee weight of $a.$ For any $\mathbf{a}=(a_1,\dots,a_n)$ in $\mathcal{R}_k^n,$ define the Lee weight of $\mathbf{a}$ as follows,
\[\wtgr(\mathbf{a})=\sum_{j=1}^n Wt_L(a_j).\]
Then we have the following result.

\begin{proposition}
If $C=\overline{\Psi}^{-1}(C_1,\dots,C_{2^k})$ is a code of length $n$ over $\mathcal{R}_k,$ then
\[d_L(C)=\min_{1\leq i\leq 2^k}d_L(C_i).\]
\label{lee}
\end{proposition}

\begin{proof}
 Let $d_L(C_j)=\min_{1\leq i\leq 2^k}d_L(C_i),$ for some $j,$ and let $\mathbf{c}_j$ be a codeword in $C_j$ such that $\wtgr(\mathbf{c}_j)=d_L(C_j).$ We have that
\[d_L(C)=\wtgr\left(\overline{\Psi}^{-1}(\mathbf{0},\dots\mathbf{0},\mathbf{c}_j,\mathbf{0},\dots,\mathbf{0})\right)=d_L(C_j).\]
\end{proof}

Since the ring $\mathcal{R}_k$ is isomorphic to $R^{2^k},$ the generating character for $\widehat{\mathcal{R}_k}$ is the product
of generating character for $\widehat{R}.$ Now, if $\chi$ is a generating character for $R,$ such that
\[\chi(x)=\xi^{wt_L(x)},\]
for any $x\in R,$ then the generating character $\chi$ for $\mathcal{R}_k$ is
\[\chi_1(\beta)=\xi^{Wt_L(\overline{\Psi}(\beta))},\]
for any $\beta\in \mathcal{R}_k.$

Define the matrix $T$ indexed by $\alpha,\beta\in \mathcal{R}_k,$ as follows
\[T_{\alpha,\beta}=\chi_\alpha(\beta)=\chi(\alpha\beta),\]
and the matrix $T_H$ as follows
\[\left(T_H\right)_{\alpha,\beta}=\chi_\alpha(\overline{\beta})=\chi(\alpha\overline{\beta}),\]
where $\overline{\beta}$ is the conjugate of $\beta$ induced by $\Theta_S,$ for some $S\subseteq\{1,2,\dots,k\}.$

Also, define the complete weight enumerator for a code $C$ as follows,
\[\cwe_C(\mathbf{X})=\sum_{\mathbf{c}\in C}\prod_{b\in \mathcal{R}_k}X_b^{n_b(\mathbf{c})},\]
where $n_b(\mathbf{c})$ is the number of occurrences of the element $b$ in $\mathbf{c}.$ Then, we have the following
result.

\begin{theorem}
If $C$ is a linear code over $\mathcal{R}_k,$ then
\begin{equation}
\cwe_{C^\bot}(\mathbf{X})=\frac{1}{|C|}\cwe_C(T\cdot \mathbf{X})
\end{equation}
and
\begin{equation}
\cwe_{C^H}(\mathbf{X})=\frac{1}{|C|}\cwe_C(T_H\cdot \mathbf{X})
\end{equation}
\label{macrel}
\end{theorem}

\begin{proof}
This theorem is a consequence of \cite[Corollary 8.2]{wood}.
\end{proof}

Note that $T$ is a $|R|^{2^k}$ by $|R|^{2^k}$ matrix indexed by the elements of $\mathcal{R}_k.$ Let $\mathcal{R}_k^\times$ be the group of units in the ring
$\mathcal{R}_k$ and let $\alpha\sim\alpha'$ if $\alpha'=u\alpha,$ for some $u\in G,$ where $G$ is a subgroup of $\mathcal{R}_k^\times.$
It can be seen that the relation $\sim$ is
an equivalence relation, so we define $\mathcal{A}=\{\alpha_1,\dots,\alpha_t\}$ be the set of representatives.
Let $S$ be the $t$ by $t$ matrix indexed by the elements in $\mathcal{A}.$ Also,
define $S_{\alpha,\beta}=\sum_{\gamma\sim\beta}T_{\alpha,\gamma}.$ We have the following lemma.

\begin{lemma}
If $\alpha\sim\alpha'$ then the row $S_{\alpha}$ is equal to the row $S_{\alpha'}.$
\end{lemma}

\begin{proof}
If $\alpha\sim\alpha'$ then for any column $\beta$ we have
\[\begin{array}{lllll}
S_{\alpha',\beta} & = & \sum_{\gamma\sim\beta}T_{\alpha',\gamma} & = & \sum_{\gamma\sim\beta}\xi^{\wtgr(\overline{\Psi}(\alpha'\gamma))}.
\end{array}\]
Since $\overline{\Psi}(\alpha\gamma)=\overline{\Psi}(\alpha)\overline{\Psi}(\gamma),$ where the multiplication in the right side of equal sign carried
out coordinate-wise, we have that
\[\begin{array}{lll}
\sum_{\gamma\sim\beta}T_{\alpha',\gamma} & = & \sum_{\gamma\sim\beta}\xi^{\wtgr(\overline{\Psi}(\alpha) \overline{\Psi}(u)\overline{\Psi}(\gamma))}\\
 & = & \sum_{\gamma'\sim\beta}\xi^{\wtgr(\overline{\Psi}(\alpha)\overline{\Psi}(\gamma'))}\\
 & = & \sum_{\gamma'\sim\beta}T_{\alpha,\gamma'}\\
 & = & S_{\alpha,\beta}.
\end{array}\]
Therefore, $S_\alpha=S_{\alpha'}$ when $\alpha\sim\alpha'.$
\end{proof}

Now, define the symmetrized weight enumerator for a code $C$ to be
\[\swe_C(\mathbf{Y_{\mathcal{A}}})=\sum_{\mathbf{c}\in C}\prod_{\alpha\in\mathcal{A}}Y_\alpha^{\swc_\alpha(\mathbf{c})},\]
where $\swc_\alpha(\mathbf{c})=\sum_{\alpha'\sim\alpha}n_{\alpha'}(\mathbf{c}).$ Then, we have the following theorem.

\begin{theorem}
If $C$ is a linear code over $\mathcal{R}_k,$ then
\[\swe_{C^\bot}=\frac{1}{|C|}\swe_C(S\cdot \mathbf{Y}_{\mathcal{A}}).\]
\label{macrel2}
\end{theorem}

\begin{proof}
Apply \cite[Theorem 8.4]{wood}.
\end{proof}

\section{Cyclic and Quasi-Cyclic Codes}

Let $C$ be a linear code of length $n$ over the ring $R.$ In this paper, we use the following definition of {\it quasi-cyclic} codes.

\begin{definition}\label{quasidef}
Let $n=md,$ for some $m$ and $d$ in $\mathbb{N}.$ Also, let $\mathbf{c}\in R^n,$ with $\mathbf{c}=\left(\mathbf{c}^{(1)}|\mathbf{c}^{(2)}|\cdots|\mathbf{c}^{(d)}\right),$ where $\mathbf{c}^{(i)}\in R^m,$ for all $i=1,2,\dots,d.$ Let $\sigma_d$ be a map from $R^n$ to $R^n$ such that $\sigma_d(\mathbf{c})=\left(\sigma\left(\mathbf{c}^{(1)}\right)|\sigma\left(\mathbf{c}^{(2)}\right)|\cdots|\sigma\left(\mathbf{c}^{(d)}\right)\right),$ where $\sigma$ is a cyclic shift from $R^m$ to $R^m.$ A code $C$ of length $n$ over ring $R$ is said to be a {\it quasi-cyclic} code with index $d$ if $\sigma_d(C)=C.$ 
\end{definition}

Note that, Definition~\ref{quasidef} is permutation equivalent to the usual definition of quasi-cyclic codes. Also, a code $C$ is said to be {\it cyclic} if its a quasi-cyclic code of index $d=1.$ We have the following characterization for quasi-cyclic codes over
the ring $\mathcal{R}_k.$

\begin{theorem}
A code $C$ of length $n$ over $\mathcal{R}_k$ is a quasi-cyclic code with index $d$ if and only if $C=\overline{\Psi}^{-1}(C_1,\dots,C_{2^k}),$
where $C_1,\dots,C_{2^k}$ are quasi-cyclic codes of length $n$ with index $d$ over $R.$
\label{quasi-cyclic}
\end{theorem}

\begin{proof}
($\Longrightarrow$) For any $i,$ take any $\mathbf{c}\in C_i.$ 
Since $C$ is a quasi-cyclic code of index $d,$, we have that
\[\overline{\Psi}^{-1}\left(\mathbf{0},\dots,\mathbf{0},\sigma_d(\mathbf{c}),\mathbf{0},\dots,\mathbf{0}\right)=\sigma_d\left(\overline{\Psi}^{-1}\left(\mathbf{0},\dots,\mathbf{0},\mathbf{c},\mathbf{0},\dots,\mathbf{0}\right)\right)\]
is also in $C.$ This gives $\sigma_d(\mathbf{c})\in C_i$ as we hope.\\[0.25cm]
($\Longleftarrow$) For any $\mathbf{w}$ in $C,$ there exist codewords $\mathbf{w}_1,\mathbf{w}_2,\dots,\mathbf{w}_{2^k},$ where $\mathbf{w}_i\in C_i,$ for all $1\leq i\leq 2^k,$ such that $\mathbf{w}=\overline{\Psi}^{-1}(\mathbf{w}_1,\dots,\mathbf{w}_{2^k}).$ Also, we have that
\[
\begin{array}{lll}
\sigma_d(\mathbf{w}) & = & \sigma_d\left(\overline{\Psi}^{-1}(\mathbf{w}_1,\dots,\mathbf{w}_{2^k})\right)\\
 & = & \overline{\Psi}^{-1}(\sigma_d(\mathbf{w}_1),\dots,\sigma_d(\mathbf{w}_{2^k})).
\end{array}
\]
Since $C_i$ is a quasi-cyclic code of index $d,$ we have $\sigma_d(\mathbf{w}_{i})$ is in $C_i,$ for all $i=1,2,\dots,2^k.$ So, $(\sigma_d(\mathbf{w}_1),\dots,\sigma_d(\mathbf{w}_{2^k}))$ is in $\overline{\Psi}(C).$ This means $\sigma_d(\mathbf{w})$ is in $C.$
\end{proof}

\begin{theorem}\label{cyclic}
A code $C$ of length $n$ over $\mathcal{R}_k$ is cyclic if and only if $C=\overline{\Psi}^{-1}(C_1,\dots,C_{2^k}),$
where $C_1,\dots,C_{2^k}$ are cyclic codes of length $n$ over $R.$
\end{theorem}

\begin{proof}
Apply Theorem~\ref{quasi-cyclic} with $d=1.$
\end{proof}

We also have the following characterization of quasi-cyclic codes.

\begin{theorem}\label{quasi-cyclic2}
A code $C$ of length $n$ over $\mathcal{R}_j$ is a quasi-cyclic code with index $d$ if and only if $\overline{\varphi}_j(C)$ is a quasi-cyclic code of length $nl_j$ with index $l_jd$ over $\mathcal{R}_{j-1}.$	
\end{theorem}

\begin{proof}
	For any $\mathbf{c}'$ in $\overline{\varphi}_j(C),$ there exists $\mathbf{c}$ in $C$ such that $\overline{\varphi}_j(\mathbf{c})=\mathbf{c}'.$ Now, let $\mathbf{c}=\left(\alpha^{(1)}|\cdots|\alpha^{(d)}\right),$ where $\alpha^{(i)}=(\alpha_{i1}+\alpha_{i1}'v_j,\dots,\alpha_{im}+\alpha_{im}'v_j),$ for all $1\leq i\leq d.$ So, we have
	\[ \begin{array}{lll}
	\mathbf{c}' & = & \overline{\varphi}_j(\mathbf{c})\\
	 & = & \left(\beta_0^{(1)}|\cdots|\beta_0^{(d)}|\beta_1^{(1)}|\cdots|\beta_1^{(d)}|\cdots|\beta_{l_j-1}^{(1)}|\cdots|\beta_{l_j-1}^{(d)}\right),
	\end{array}\] 
	where $\beta_0^{(i)}=(\alpha_{i1},\dots,\alpha_{im}),$ for all $1\leq i\leq d,$ and 
	\[\beta_{r}^{(i)}=(\beta_r\alpha_{i1}+\beta_r'\alpha_{i1}',\dots,\beta_r\alpha_{im}+\beta_r'\alpha_{im}'),\]
	for all $r=1,\dots,l_j-1,$ $i=1,\dots,d.$ Consider,
	\[\begin{array}{lll}
	\overline{\varphi}_j(\sigma_d(\mathbf{c})) & = & \left(\sigma\left(\beta_0^{(1)}\right)|\cdots|\sigma\left(\beta_0^{(d)}\right)|\sigma\left(\beta_1^{(1)}\right)|\cdots|\sigma\left(\beta_1^{(d)}\right)|\cdots\right.\\
	& & \left.\cdots|\sigma\left(\beta_{l_j-1}^{(1)}\right)|\cdots|\sigma\left(\beta_{l_j-1}^{(d)}\right)\right)\\
	& = & \sigma_{l_jd}(\mathbf{c}').
	\end{array}\]
	Therefore, $\sigma_d(\mathbf{c})\in C$ if and only if $\sigma_{l_jd}(\mathbf{c}')\in \overline{\varphi}_j(C).$ 
\end{proof}

The following results are direct consequences of Theorem~\ref{quasi-cyclic2}.

\begin{theorem}
	A code $C$ of length $n$ over $\mathcal{R}_j$ is a cyclic code if and only if $\overline{\varphi}_j(C)$ is a quasi-cyclic code of length $nl_j$ with index $l_j$ over $\mathcal{R}_{j-1}.$
\end{theorem}

\begin{corollary}
	A code $C$ of length $n$ over $\mathcal{R}_k$ is a quasi-cyclic code with index $d$ if and only if $\overline{\varphi}_1\circ\cdots\circ\overline{\varphi}_k(C)$ is a quasi-cyclic code of length $nl_1\cdots l_k$ with index $d\cdot l_1\cdots l_k$ over $R.$
\end{corollary}

\begin{proof}
	Apply Theorem~\ref{quasi-cyclic2} repeatedly while considering the image of $\overline{\varphi}_1\circ\cdots\circ\overline{\varphi}_k.$
\end{proof}

\begin{corollary}
	A code $C$ of length $n$ over $\mathcal{R}_k$ is a cyclic code if and only if $\overline{\varphi}_1\circ\cdots\circ\overline{\varphi}_k(C)$ is a quasi-cyclic code of length $nl_1\cdots l_k$ with index $l_1\cdots l_k$ over $R.$
\end{corollary}

\section{Skew-Cyclic and Quasi-Skew-Cyclic Codes}

Let $C$ be a code of length $n$ over the ring $\mathcal{R}_k.$ Given an atomorphism on the ring $\mathcal{R}_k,$ say $\theta,$ then $C$ is said to be a $\theta$-cyclic
code or skew-cyclic code if
\begin{itemize}
 \item[(1)] $C$ is a linear code over $\mathcal{R}_k,$ and
 \item[(2)] For any $c=(c_0,\dots,c_{n-1})$ in $C,$ we have that $T_\theta(c)=\left(\theta(c_{n-1}),\theta(c_{0}),\dots,\theta(c_{n-2})\right)$ is also in $C.$
\end{itemize}

Also, $C$ is said to be a quasi-$\theta$-cyclic code of index $d$ if
\begin{itemize}
 \item[(1)] $C$ is a linear code over $\mathcal{R}_k,$ and
 \item[(2)] For any $c=(c_0,\dots,c_{n-1})$ in $C,$ we have that $T^d_\theta(c)=\left(\theta(c_{n-d}),\theta(c_{n-d+1}),\dots,\theta(c_{n-d-1})\right)$
	    is also in $C.$
\end{itemize}

Let $T$ be a cyclic-shift operator on $R^{n2^k}.$ We have the following characterizations.

\begin{theorem}
A code $C$ over $\mathcal{R}_k$ is a quasi-$\theta$-cyclic code of index $d$ if and only if
$T^{d2^k}\circ\Sigma_S\circ \Phi_{S_1,S_2}(\overline{\Psi}(C))\subseteq \overline{\Psi}(C),$
for some $S,S_1,S_2\subseteq \{1,2,\dots,k\},$ where $|S_1|=|S_2|.$
\label{quasi-skew}
\end{theorem}

\begin{proof}
Let $c=(c_0,c_1,\dots,c_{n-1})$ be any element in $C.$ We can see that
\[\overline{\Psi}(c_{n-d},c_{n-d+1},\dots,c_{n-d-1})=T^{d2^k}(\overline{\Psi}(c_0,\dots,c_{n-1})).\]
Since $\theta$ is a composition of $\Theta_S$ and $\Phi_{S_1,S_2},$ for some $S,S_1,S_2\subseteq\{1,2,\dots,k\},$ we have that
\[\overline{\Psi}(T^d_\theta(c))=T^{d2^k}\left(\Sigma_S\left(\Gamma_{S_1,S_2}\left(\overline{\Psi}(c)\right)\right)\right).\]
Therefore,  $C$ is invariant under the action of $T^d_\theta$ if and only if $\overline{\Psi}(C)$ invariant under the action of
$T^{d2^k}\circ\Sigma_S\circ\Gamma_{S_1,S_2}.$
\end{proof}

\begin{theorem}
A code $C$ over $\mathcal{R}_k$ is a $\theta$-cyclic code if and only if $T^{2^k}\circ\Sigma_S\circ \Phi_{S_1,S_2}(\overline{\Psi}(C))\subseteq \overline{\Psi}(C),$
for some $S,S_1,S_2\subseteq \{1,2,\dots,k\},$ where $|S_1|=|S_2|.$
\label{skew}
\end{theorem}

\begin{proof}
Apply Theorem~\ref{quasi-skew} with $d=1.$
\end{proof}

We can also have more technical characterizations as follow.

\begin{theorem}
A linear code $C$ over $\mathcal{R}_k$ is quasi-$\theta$-cyclic of index $d$ and length $n$ if and only if
there exist quasi-$\vartheta$-cyclic codes $C_1,C_2,\dots,C_{2^k}$
of length $n$ over $R$
with index $d\cdot\ord(\tilam),$ such that
\[
C=\overline{\Psi}_k^{-1}(C_1,C_2,\dots, C_{2^k})
\]
where $\vartheta$ is an automorphism in $R,$ and
$T^d_{\tilde{\theta}}(C_i)\subseteq C_j,$ where $j\in S\cup S_2,$ for all
$i=1,2,\dots,2^k.$
\label{quasi-skew-2}
\end{theorem}

\begin{proof}
$(\Longrightarrow)$ Remember that there exist codes over $R,$ $C_1,C_2,\dots,C_{2^k},$ such that,
\[C=\overline{\Psi}_k^{-1}(C_1,C_2,\dots,C_{2^k}).\]
For any $c_i\in C_i,$ let $c_i=(\alpha_1,\dots, \alpha_{n}).$ If,
$c=\overline{\Psi}_k^{-1}(0,\dots,0,c_i,0,\dots,0),$ then
\[
\left(\alpha_1v_{S_i}-\sum_{A\supsetneq S_i}\alpha_1v_A,\dots,
\alpha_nv_{S_i}-\sum_{A\supsetneq S_i}\alpha_nv_A\right).
\]
So, if we consider
\[
\overline{\Psi}_k(T_{\theta}^{dt_1}(c))=(0,\dots,0,T_{\vartheta}^{dt_1}(c_i),0,\dots,0),
\]
then we have
$T_{\vartheta}^d(c_i)$ is in $C_j,$ where $j\in S\cup S_2.$
By continuing this process, we have $T^{d\cdot\ord(\tilam)}_{\vartheta}(c_i)\in C_i,$
which means, $C_i$ is quasi-$\vartheta$-cyclic code over $R$
with index $d\cdot\ord(\tilam),$ for all
$i=1,\dots,2^k.$\\

$(\Longleftarrow)$ For any $c\in C,$ we can see that
$\overline{\Psi}_k(c)\in (C_1,\dots,C_{2^k}).$
Since $C_i$ is quasi-$\vartheta$-cyclic
code over $R$
with index $d\cdot\ord(\tilam),$ for all
$i=1,\dots,2^k,$ $C_1,$
%and $C_{2^k}$ are $\tilde{\theta}$-cyclic codes,
and $T_{\vartheta}^{dt_1}(C_i)\subseteq C_j,$ where $j\in S\cup S_2,$ for all
$i=1,2,\dots,2^k,$ where $1\leq t_1\leq 2^k.$
Then we have  $T_{\theta}^d(c)=\overline{\Psi}_k^{-1}(T_{\vartheta}(\Psi_k(c)))\in C,$ as we hope.
\end{proof}

\begin{theorem}
A linear code $C$ over $\mathcal{R}_k$ is $\theta$-cyclic of length $n$ if and only if
there exist quasi-$\vartheta$-cyclic codes $C_1,C_2,\dots,C_{2^k}$
of length $n$ over $R$
with index $\ord(\tilam),$ such that
\[
C=\overline{\Psi}_k^{-1}(C_1,C_2,\dots, C_{2^k})
\]
where $\vartheta$ is an automorphism in $R,$ and
$T_{\tilde{\theta}}(C_i)\subseteq C_j,$ where $j\in S\cup S_2,$ for all
$i=1,2,\dots,2^k.$
\label{skew-2}
\end{theorem}

\begin{proof}
Apply Theorem~\ref{quasi-skew-2} with $d=1.$
\end{proof}

Theorem~\ref{quasi-skew-2} gives us an algorithm to construct quasi-skew-cyclic codes over the ring $B_k$ as follows.
\vspace{0.6cm}
\begin{alg}
Given $n, d,$ the ring $\mathcal{R}_k,$ and an automorphism $\theta.$

\begin{itemize}
 \item[(1)] Decompose $\theta$ to be $\theta=\Theta_{S}\circ\Phi_{S_1,S_2}.$

\item[(2)] Determine $\ord(\tilam)$ and $\vartheta.$

\item[(3)] Choose quasi-$\vartheta$-cyclic codes over $R,$ say $C_1,\dots,C_{2^k},$ such that
\[
T_{\tilde{\theta}}^{dt_1}(C_i)\subseteq C_{j},
\]where $j\in S\cup S_2,$ for all $i=1,2,\dots,2^k.$

\item[(4)] Calculate $C=\overline{\Psi}_k^{-1}(C_1,\dots,C_{2^k}).$

\item[(5)] $C$ is a quasi-$\theta$-cyclic code of index $d$ over the ring $\mathcal{R}_k.$
\end{itemize}
\label{algorithm}
\end{alg}
Note that Algorithm~\ref{algorithm} can be used to construct skew-cyclic code over $\mathcal{R}_k$ by choosing $d=1.$

\section{Examples}
\subsection{Examples using the map $\Psi$}

% As we know that, $Aut(\mathbb{Z}_m)\cong \mathbb{Z}_m^\times,$ {\it i.e.} the automorphisms group of $\mathbb{Z}_m$ is isomorphic to the group of units in $\mathbb{Z}_m.$ Therefore, $|Aut(\mathbb{Z}_m)|=\varphi(m),$ where $\varphi$ is the Euler totient function, and the automorphisms in the ring $\mathcal{R}_k$ are the compositions between the map $\Theta_S,$ for some $S\subseteq\{1,2,\dots,k\},$ and the map $\Phi_{S_1,S_2},$ for some $S_1,S_2\subseteq\{1,2,\dots,k\}$ and $|S_1|=|S_2|.$ We have to note that the map $\Phi_{S_1,S_2}$ is of the form
%
%\[\Phi_{S_1,S_2}(\alpha v_j)= \vartheta(\alpha)v_{\phi_{S_1,S_2}(j)},\]
%where $\vartheta(\alpha)=\beta\alpha,$ for some $\beta\in\mathbb{Z}_m^\times.$ \\[0.3cm]

As a direct consequence of Theorem~\ref{linearpsi}, we have that for any code $C$ of length $n$ over $\mathcal{R}_k=\mathbb{Z}_m[v_1,v_2,\dots,v_k],$ where $v_i^2=v_i,$ for all $i=1,2,\dots,k,$ there exist codes $C_1,C_2,\dots,C_{2^k}$ of length $n$ over $\mathbb{Z}_m$ such that $C=\overline{\Psi}^{-1}(C_1,C_2,\dots,C_{2^k}).$ 

\begin{example}
Let $\mathcal{R}_1=\mathbb{Z}_4[v],$ where $v^2=v.$ Also, let $C=\langle (1\;v\;1+v\;3)\rangle.$ We can check that
\[\overline{\Psi}((1\;v\;1+v\;3))=\left(\begin{array}{cccc}
1 & 0 & 1 & 3 \\
1 & 1 & 2 & 3
\end{array}\right).\]
Then, if we choose $C_1=\langle (1\;0\;1\;3)\rangle$ and $C_2=\langle (1\;1\;2\;3)\rangle,$ we have $C=\overline{\Psi}^{-1}(C_1,C_2).$
\end{example}

Moreover, we can have more explicit example for Hermitian self-dual codes as follow.

\begin{example}
Let $\mathcal{R}_1=\mathbb{Z}_4[v],$ where $v^2=v.$ In this ring, $\Theta_1(v)=1-v.$ Let $C=\langle(v\;v\;v)\rangle$ be a code over $\mathcal{R}_1.$ By Proposition~\ref{hermitianconstruct}, $C$ is a Hermitian self-dual code. Since
\[\overline{\Psi}((v\;v\;v))=\left(\begin{array}{ccc}
1 & 1 & 1  \\
1 & 1 & 1
\end{array}\right),\] we have that $C=\overline{\Psi}^{-1}(C_1,C_2),$ where $C_1=C_2=\langle (1\;1\;1)\rangle.$ We can check that $C_1$ is an Euclidean self-dual code over $\mathbb{Z}_4.$ Therefore, we have $C_2=C_1^\perp,$ as stated in Proposition~\ref{hermitian} and Theorem~\ref{hermitian2}.
\end{example}

Also, we have the following example for Euclidean self-dual codes.

\begin{example}
Let $\mathcal{R}_1=\mathbb{Z}_4[v],$ where $v^2=v.$ Take $C=\langle (v\; 1-v),(1-v\; v)\rangle.$ We can see that $C$ is an Euclidean self-dual code over $\mathcal{R}_1.$ Also, we know that
\[\overline{\Psi}((v\;1-v))=\left(\begin{array}{cc}
0 & 1  \\
1 & 0
\end{array}\right),\]
and
\[\overline{\Psi}((1-v\;v))=\left(\begin{array}{cc}
1 & 0  \\
0 & 1
\end{array}\right).\] If we take $C_1=C_2=\langle(1\;0),(0\;1)\rangle,$ then we have $C=\overline{\Psi}^{-1}(C_1,C_2).$ We can check that $C_1$ and $C_2$ are Euclidean self-dual codes over $\mathbb{Z}_4$ also, as stated in Theorem~\ref{euclid}.
\end{example}

\subsection{Codes over $\mathbb{Z}_4$}

In this part, we will use the map $\varphi_1$ to get codes over $\mathbb{Z}_4$ from codes over $\mathcal{R}_1=\mathbb{Z}_4+v\mathbb{Z}_4,$ where $v^2=v.$ For any element $\mathbf{x}=(x_1,\dots,x_n)$ in  $\mathbb{Z}_4^n,$ Lee weight of $\mathbf{x},$ denoted by $w_L(\mathbf{x}),$ as

\begin{equation}
w_L(\mathbf{x})=\sum_{i=1}^n \min\{|x_i|,|4-x_i|\}.
\end{equation} 
Using the above weight, we define Lee distance of a code $C$ as 
\[d_L(C)=\min_{\mathbf{c}\in C\atop \mathbf{c}\not=\mathbf{0}}w_L(\mathbf{c}).\]
 We will give some examples of codes over $\mathbb{Z}_4$ with maximum Lee distance so far, as in \url{http://www.asamov.com/Z4Codes/CODES/ShowCODESTablePage.aspx}, constructed using the map $\varphi_1.$  

\begin{example}\label{exz41}
	Define a map $\varphi_1$ as follows.
	\[\begin{array}{llll}
		\varphi_1 & \mathbb{Z}_4+v\mathbb{Z}_4 & \longrightarrow & \mathbb{Z}_4^2\\
				  & \alpha+v\beta & \longmapsto & (\alpha,2\alpha+\beta).
	\end{array}\]
	Let $C=\langle 1+v\rangle=\{0, 1+v, 2+2v, 3+3v, 2v, 2, 1+3v, 3+v\}$ be a code of length 1 over $\mathcal{R}_1=\mathbb{Z}_4+v\mathbb{Z}_4,$ where $v^2=v.$ We have, 
	\[\varphi_1(1+v)=(1,3),\quad \varphi_1(2+2v)=(2,2),\quad \varphi_1(3+3v)=(3,1),\quad \varphi_1(2v)=(0,2),\]
	\[\varphi_1(2)=(2,0),\quad \varphi_1(1+3v)=(1,1),\quad \varphi_1(3+v)=(3,3).\]
	We can see that $d_L(\varphi_1(C))=2$ and $|\varphi_1(C)|=8.$ 
\end{example}

\begin{example}\label{exz42}
	Define a map $\varphi_1$ as follows.
	\[\begin{array}{llll}
	\varphi_1 & \mathbb{Z}_4+v\mathbb{Z}_4 & \longrightarrow & \mathbb{Z}_4^3\\
	& \alpha+v\beta & \longmapsto & (\alpha,\beta,\alpha+\beta).
	\end{array}\]
	Let $C=\langle 2\rangle=\{0, 2, 2v, 2+2v\}.$ We have that
	\[\varphi_1(2)=(2,0,2),\quad \varphi_1(2v)=(0,2,2),\quad \varphi_1(2+2v)=(2,2,0).\]
	So, $d_L(\varphi_1(C))=4$ and $|\varphi_1(C)|=4.$
\end{example}

\begin{example}\label{exz43}
	Define a map $\varphi_1$ as follows.
	\[\begin{array}{llll}
	\varphi_1 & \mathbb{Z}_4+v\mathbb{Z}_4 & \longrightarrow & \mathbb{Z}_4^5\\
	& \alpha+v\beta & \longmapsto & (\alpha,\beta,\alpha+\beta,\alpha,\alpha+\beta).
	\end{array}\]
	Let $C=\langle 2\rangle.$ We can see that,
	\[\varphi_1(2)=(2,0,2,0,2),\quad \varphi_1(2v)=(0,2,2,0,2),\quad \varphi_1(2+2v)=(2,2,0,2,0).\]
	Therefore, we have $d_L(\varphi_1(C))=6$ and $|\varphi_1(C)|=4.$
\end{example}
The following table gives some examples of codes over $\mathbb{Z}_4$ obtained by a similar way as in Example~\ref{exz41}-\ref{exz43}.

\begin{table}[h]
	\centering
	\begin{tabular}{|ccccc|}
	 \hline\hline
	 $n$ & $C$ & $\varphi_1$ & $d_L(\varphi_1(C))$ & $|\varphi_1(C)|$ \\ [0.5ex]
	 \hline\hline
	 2 & $\langle 1+v\rangle$ & $\alpha+v\beta\mapsto (\alpha,2\alpha+\beta)$ & 2 & 8\\
	 2 & $\langle 2\rangle$ & $\alpha+v\beta\mapsto (\alpha,\alpha+\beta)$ & 2 & 4 \\
	 3 & $\langle 2\rangle$ & $\alpha+v\beta\mapsto (\alpha,\beta,\alpha+\beta)$ & 4 & 4 \\
	 3 & $\langle 2+2v\rangle$ & $\alpha+v\beta\mapsto (\alpha,\beta,\alpha+\beta)$ & 4 & 2 \\
	 3 & $\langle 2v\rangle$ & $\alpha+v\beta\mapsto (\alpha,\beta,\alpha+\beta)$ & 4 & 2 \\
	 4 & $\langle 2v\rangle$ & $\alpha+v\beta\mapsto (\alpha,\beta,\alpha+\beta,\alpha+\beta)$ & 6 & 2 \\
	 5 & $\langle 2\rangle$ & $\alpha+v\beta\mapsto (\alpha,\beta,\alpha+\beta,\alpha,\alpha+\beta)$ & 6 & 4 \\
	 \hline\hline
	\end{tabular}
\caption{Some examples of codes over $\mathbb{Z}_4.$}
\end{table}

\end{document}